\documentclass[letterpaper]{article}
\usepackage{natbib}

\usepackage{amsmath,amsfonts,amsthm,amssymb}
\usepackage{enumerate}
\usepackage{relsize}
\usepackage{algpseudocode}
\usepackage{algorithm}

\usepackage{tikz}
\usetikzlibrary{arrows}
\usetikzlibrary{positioning}
\usetikzlibrary{backgrounds}
\usetikzlibrary{automata}

\newcommand{\R}{\mathbb{R}}
\renewcommand{\citename}[1]{\citeauthor{#1}~[\citeyear{#1}]}

\newcommand{\Z}{\mathbb{Z}}

\newcommand{\NP}{\mathit{NP}}
\newcommand{\RP}{\mathit{RP}}
\newcommand{\pseudodim}{\mathit{Pdim}}
\newcommand{\vcdim}{\mathit{VCdim}}

\renewcommand{\cal}[1]{\mathcal{#1}}
\newcommand{\tup}[1]{\langle #1 \rangle}

\newcommand{\eps}{\varepsilon}
\newcommand{\maxflow}{\mathit{flow}}

\newcommand{\euler}{\mathrm{e}}

\newcommand{\CoS}{\mathit{CoS}}

\newcommand{\core}{\mathit{Core}}
\renewcommand{\vec}[1]{\mathbf{#1}}

\newcommand{\TTG}{\cal C_{\mathit{ttg}}}
\newtheorem{theorem}{Theorem}[section]
\newtheorem{lemma}[theorem]{Lemma}

\theoremstyle{definition}

\title{Learning Cooperative Games}
\author{
Maria-Florina Balcan, Ariel D. Procaccia and Yair Zick\\
Carnegie Mellon University, National University of Singapore\\
\texttt{ninamf,arielpro@cs.cmu.edu, zick@comp.nus.edu.sg}
}
\date{}
\begin{document}

\maketitle

\begin{abstract}
This paper explores a PAC (probably approximately correct) learning model in cooperative games. Specifically, we are given $m$ random samples of coalitions and their values, taken from some unknown cooperative game; can we predict the values of unseen coalitions? 
We study the PAC learnability of several well-known classes of cooperative games, such as network flow games, threshold task games, and induced subgraph games. We also establish a novel connection between PAC learnability and core stability: for games that are efficiently learnable, it is possible to find payoff divisions that are likely to be stable using a polynomial number of samples. 
\end{abstract}
\section{Introduction}
Cooperative game theory studies the following model. We are given a set of players $N = \{1,\dots,n\}$, and $v:2^N \to \R$ is a function assigning a value to every subset (also referred to as a {\em coalition}) $S \subseteq N$. 

The game-theoretic literature generally focuses on revenue division: suppose that players have formed the coalition $N$, they must now divide the revenue $v(N)$ among themselves in some reasonable manner. 
However, all of the standard solution concepts for cooperative games require intimate knowledge of the structure of the underlying coalitional interactions. For example, suppose that a department head wishes to divide company bonuses among her employees in a canonically stable manner using the \emph{core} --- a division such each coalition is paid (in total) at least its value. In order to do so, she must know the value that would have been generated by every single subset of her staff. How would she obtain all this information?

Indeed, it is the authors' opinion that the \emph{information} required in order to compute cooperative solution concepts (much more than computational complexity) is a major obstacle to their widespread implementation.  

Let us therefore relax our requirements. Instead of querying every single coalition value, we would like to elicit the underlying structure of coalitional interactions using a sample of $m$ evaluations of $v$ on subsets of $N$. To be more specific, let us focus on the most common learning-theoretic model: the {\em probably approximately correct (PAC) model}~\citep{kearns1994introduction}. Briefly, the PAC model studies the following problem: we are given a set of points $\vec x_1,\dots,\vec x_m\in \R^n$ and their values $y_1,\dots,y_m$. There is some function $f$ that generated these values, but it is not known to us. We are interested in finding a function $f^*$ that, given that $\vec x_1,\dots,\vec x_m$ were independently sampled from some distribution $\cal D$, is \emph{very likely} (``probably'') to agree with $f$ on \emph{most} (``approximately'') points sampled from the same distribution.

\citename{procaccia2006learning} provide some preliminary results on PAC learning cooperative games, focusing on \emph{simple} games (this is a technical term, not an opinion!) --- where $v(S)\in \{0,1\}$ for every $S\subseteq N$. 
Their results are mostly negative, showing that simple games require an exponential number of samples in order to be properly PAC learned (with the exception of the trivial class of unanimity games). However, the decade following the publication of their work has seen an explosive growth in the number of well-understood classes of cooperative games, as well as a better understanding of the computational difficulties one faces when computing cooperative solution concepts. This is where our work comes in.

\subsection{Our Contribution}

We revisit the connection between learning theory and cooperative games, greatly expanding on the results of \citename{procaccia2006learning}. 

In Section~\ref{sec:probablystable}, we introduce a novel relaxation of the core: it is likely (but, in contrast to the classic core, not certain) that a coalition cannot improve its payoff by working alone. 
Focusing on probable stability against likely deviations saves us a lot of computational overhead: our first result (Theorem~\ref{thm:pac-stable}) shows that {\em any} cooperative game is PAC stabilizable; that is, there exists an algorithm that will output a payoff division that is likely to be resistant against future deviations by sets sampled from a distribution $\cal D$, given a polynomial number of sets sampled i.i.d. from the same distribution. What's more, the payoff outputted is feasible: it is no more than $v(N)$ if the core of the game is not empty; if the core of $v$ is empty, then the total payoff will be no more than the minimum required in order to stabilize the game. In other words, this algorithm will pay no more than the {\em cost of stability}~\cite{costab2009} of the game $v$.

While coalitional stability is naturally desirable, understanding the underlying coalitional dynamics is no less important. In Section~\ref{sec:common} we ask whether or not classes of games are \emph{efficiently learnable}; that is, is there a polynomial-time algorithm that receives a polynomial number of samples, and outputs an accurate hypothesis with high confidence. Our main results are that network flow games~\citep[Chapter 17.9]{gtbook} are efficiently learnable with path queries (but not in general), and so are threshold task games~\citep{ocfgeb}, and induced subgraph games~\citep{deng94complexity}. We also study  $k$-vector weighted voting games~\citep{elkind2009computational}, MC nets~\citep{ieong2005mcnets}, and coalitional skill games~\citep{bachrach2008skill}.

\subsection{Related Work}
Aside from the closely related work of \citename{procaccia2006learning}, there are several papers that study coalitional stability in uncertain environments. \citename{chalkiadakis2004bayesian} and \citename{li2015noise} assume that coalition values are drawn from some unknown distribution, and we observe noisy estimates of the values. However, both papers assume full access to the cooperative game, whereas we assume that $m$ independent samples are observed. Other works study coalitional uncertainty: coalition values are known, but agent participation is uncertain due to failures~\citep{bachrach2012agent,bachrach2012solving,bachrach2013reliability}.  

Our work is also related to papers on eliciting and learning combinatorial valuation functions~\citep{zinkevich2003preference,lahaie2004applying,lahaie2005more,balcan2011submodular,balcan2011learning,badanidiyuru2012sketching}. A player's valuation function in a combinatorial auction is similar to a cooperative game: it assign a value to every subset of items (instead of every subset of players). This connection allows us to draw on some of the insights from these papers. For example, as we explain below, learnability results for XOS valuations~\citep{balcan2011learning} informed our results on network flow games. 

\section{Preliminaries}
\subsection{Cooperative Games}
A cooperative game is a tuple $\cal G =\tup{N,v}$, where $N = \{1,\dots,n\}$ is a set of players, and $v:2^n \to \R$ is called the {\em characteristic function} of $\cal G$. 
When the player set $N$ is obvious, we will identify $\cal G$ with the characteristic function $v$, referring to $v$ as the game. A game $\cal G$ is called {\em simple} if $v(S) \in \{0,1\}$ for all $S \subseteq N$; $\cal G$ is called {\em monotone} if $v(S) \le v(T)$ whenever $S\subseteq T$. 
One of the main objectives in cooperative games is finding ``good'' ways of dividing revenue: it is assumed that players have generated the revenue $v(N)$, and must find a way of splitting it. 
An {\em imputation} for $\cal G$ is a vector $\vec x \in \R^n$ that satisfies {\em efficiency}: $\sum_{i = 1}^n x_i = v(N)$, and {\em individual rationality}: $x_i \ge v(\{i\})$ for every $i \in N$. The set of imputations, denoted $I(\cal G)$, is the set of all possible ``reasonable'' payoff divisions among the players. 
Given a game $\cal G$, the {\em core} of $\cal G$ is given by
$$\core(\cal G) = \{\vec x \in I(\cal G)\mid \forall S \subseteq N: x(S) \ge v(S)\}.$$
The core is the set of all {\em stable} imputations: no subset of players $S$ can deviate from an imputation $\vec x \in \core(\cal G)$ while guaranteeing that every $i \in S$ receives at least as much as it gets under $\vec x$.

\subsection{PAC Learning}
We provide a brief overview of the PAC learning model; for a far more detailed exposition, see~\citep{kearns1994introduction,shashua2009ml}. PAC learning pertains to the study of the following problem: we are interested in learning an unknown function $f:2^N \to \R$. In order to estimate the value of $f$, we are given $m$ samples $(S_1,v_1),\dots,(S_m,v_m)$, where $v_j = f(S_j)$. 
Without any additional information, one could make arbitrary guesses as to the possible identity of $f$; for example, we could very well guess that $f^*(S_j) = v_j$ for all $j \in [m]$, and 0 everywhere else. 
Thus, in order to obtain meaningful results, we must make further assumptions. First, we restrict $f$ to be a function from a certain class of functions $\cal C$: for example, we may know that $f$ is a linear function of the form $f(S) = \sum_{i \in S}w_i$, but we do not know the values $w_1,\dots,w_n$. Second, we assume that there is some distribution $\cal D$ over $2^N$ such that $S_1,\dots,S_m$ were sampled i.i.d. from $\cal D$. Finally, we require that the estimate that we provide has low error over sets sampled from $\cal D$. 

Formally, we are given a function $v:2^N \to \R_+$, and two values $\eps>0$ (the accuracy parameter) and $\delta > 0$ (the confidence parameter). An algorithm $\cal A$ takes as input $\eps$, $\delta$ and $m$ samples, $(S_1,v(S_1)),\dots,(S_m,v(S_m))$, taken i.i.d. from a distribution $\cal D$. 
We say that $\cal A$ can properly learn a function $f\in \cal C$ from a class of functions $\cal C$ ($\cal C$ is sometimes referred to as the \emph{hypothesis class}), if by observing $m$ samples --- where $m$ can depend only on $n$ (the representation size), $\frac 1 \eps$, and $\frac 1 \delta$ --- it outputs a function $f^* \in \cal C$ such that with probability at least $1 - \delta$,
$$\Pr_{S \sim \cal D}[f(S) \ne f^*(S)] < \eps.$$
The confidence parameter $\delta$ indicates that there is some chance that $\cal A$ will output a bad guess (intuitively, that the $m$ samples given to the algorithm are not representative of the overall behavior of $f$ over the distribution $\cal D$), but this is unlikely. The accuracy parameter $\eps$ indicates that for most sets sampled from $\cal D$, $f^*$ will correctly guess the value of $S$. 

Note that the algorithm $\cal A$ does not know $\cal D$; that is, the only thing required for PAC learnability to hold is that the input samples independent, and that future observations are also sampled from $\cal D$. In this paper, we only discuss {\em proper learning}; that is, learning a function $f \in \cal C$ using only functions from $\cal C$. 

We say that a finite class of functions $\cal C$ is {\em efficiently PAC learnable} if the PAC learning algorithm described above runs in polynomial time, and its sample complexity $m$ is polynomial in $n$, $\frac 1 \epsilon$, and $\frac 1 \delta$.

Efficient PAC learnability can be established via the existence of {\em consistent algorithms}. Given a class of functions $\cal C$ from $2^N$ to $\R$, suppose that there is some efficient algorithm $\cal A$ that for any set of samples $(S_j,v_j)_{j = 1}^m$ is able to output a function $f^* \in \cal C$ such that $f^*(S_j) = v_j$ for all $j \in [m]$, or determine that no such function exists. Then $\cal A$ is an algorithm that can efficiently PAC learn $\cal C$ given $m \ge \frac1\eps \log\frac{|\cal C|}{\delta}$ samples. Conversely, if no efficient algorithm exists, then $f$ cannot be efficiently PAC learned from $\cal C$.

To conclude, in order for a class $\cal C$ to be efficiently PAC learnable, we must have polynomial bounds on the {\em sample complexity} --- i.e. the number of samples required in order to obtain a good estimate of functions in $\cal C$ --- as well as a poly-time algorithm that finds a function in $\cal C$ which is a perfect match for the samples. We observe that in many of the settings described in this paper, the sample complexity is low, but finding consistent functions in $\cal C$ is computationally intractable (it would entail that $P = \NP$ or that $\NP = \RP$). In contrast, the result of \citename{procaccia2006learning} establishes lower bounds on the sample complexity for PAC learning monotone simple games, but there exists a simple algorithm that outputs a hypothesis consistent with any sample.

When the hypothesis class $\cal C$ is finite, it suffices to show that $\log|\cal C|$ is bounded by a polynomial in order to establish a polynomial sample complexity. In the case of an infinite class of hypotheses, this bound becomes meaningless, and other measures must be used. When learning a function that takes values in $\{0,1\}$, the VC dimension~\citep{kearns1994introduction} captures the learnability of $\cal C$. Given a class $\cal C$, and a list $\cal S$ of $m$ sets $S_1,\dots,S_m$, we say that $\cal C$ {\em shatters} $\cal S$ if for every $\vec b \in \{0,1\}^m$ there exists some $v_{\vec b} \in \cal C$ such that $v(\vec b)(S_j) = b_j$ for all $j$. We write $$\vcdim(\cal C) = \max\{m\mid \exists \cal S,|\cal S| = m,\cal C \mbox{ can shatter }\cal S\}.$$ 

When learning hypotheses that output real numbers (as opposed to functions that take on values in $\{0,1\}$), the notion of {\em pseudo dimension} is used in order to bound the complexity of a function class. 
Given a sample of $m$ sets $\cal S = S_1,\dots,S_m \subseteq N$, we say that a class $\cal C$ {\em shatters} $\cal S$ if there exist thresholds $r_1,\ldots,r_m\in \R$ such that for every $\vec b \in \{0,1\}^m$ there exists some $v_{\vec b} \in \cal C$ such that $v_{\vec b}(S_j) \ge r_j$ if $b_j = 1$, and $v_{\vec b}(S_j) < r_j$ if $b_j = 0$. We write 
$$\pseudodim(\cal C) = \max\{m \mid \exists \cal S: |\cal S| = m,\cal C \mbox{ can shatter }\cal S\}.$$
It is known~\citep{anthony2009neural} that if $\pseudodim(\cal C)$ is polynomial, then the sample complexity of $\cal C$ is polynomial as well.

\section{PAC Stability}\label{sec:probablystable}
In the context of cooperative games, one could think of PAC learning as the following process. A central authority wishes to find a stable outcome, but lacks information about agents' abilities. It solicits the independent valuations of $m$ subsets of agents, and outputs an outcome that, with probability $1 - \delta$, is likely to be stable against any unknown valuations. 

More formally, given $\eps \in (0,1)$, we say that an imputation $\vec x \in I(\cal G)$ is {\em $\eps$-probably stable} under $\cal D$ if 

$$\Pr_{S\sim \cal D}\left[x(S) \ge v(S) \right] \ge 1 - \eps.$$

An algorithm $\cal A$ can {\em PAC stabilize} a class of functions $\cal C$ from $2^N$ to $\R$ if, given $\eps,\delta\in (0,1)$, and $m$ i.i.d. samples $(S_1,v(S_1)),\dots,(S_m,v(S_m))$ of some $v \in \cal C$, with probability $1 - \delta$, $\cal A$ outputs an outcome $\vec x$ that is $\eps$-probably stable under $\cal D$. 
There is an important subtlety here: suppose that we know the value $v(N)$; by grossly overpaying the agents we could easily stabilize any game. Under the mild assumption of monotonicity, we can pay each $i \in N$ a value of $v(N)$, which results in a trivially stable payoff division. 
In addition to PAC stability, we must ensure that the total payment to agents is no more than $v(N)$ if the core is not empty; if the core of $\cal G$ is not empty, then the output of our algorithm should be no more than the minimal payment required in order to stabilize the game. Formally, given a cooperative game $\cal G = \tup{N,v}$ and a non-negative constant $\Delta \in \R_+$, we define $\cal G_\Delta = \tup{N,v_\Delta}$, where $v_\Delta(S) = v(S)$ for all $S \subset N$, and $v_\Delta(N) = v(N) + \Delta$. The {\em cost of stability}~\cite{costab2009} is 
$$\CoS(\cal G) = \min\{\Delta \in \R_+ \mid \core(\cal G_\Delta) \ne \emptyset\}.$$
If $\core(\cal G) \ne \empty$ then $\CoS(\cal G) = 0$; indeed, the larger the value of $\CoS(\cal G)$, the greater the subsidy required in order to stabilize $\cal G$. The following theorem establishes the poly-time PAC stabilizability of cooperative games, with formal bounds on the total payoff provided by the outputted imputation\footnote{The authors would like to thank Amit Daniely for pointing out some of the basic ideas outlined in the proof of Theorem~\ref{thm:pac-stable}.}. 
\begin{theorem}\label{thm:pac-stable}
There exists an algorithm that, given a cooperative game $\cal G = \tup{N,v}$ and $m$ i.i.d. samples $(S_1,v(S_1)),\dots,(S_m,v(S_m))$ sampled from a distribution $\cal D$, where $m$ is polynomial in $\frac1\eps,\log\frac1\delta$, outputs a payoff division $\vec x^*$ that $(\eps,\delta)$ PAC stabilizes $\cal G$. Furthermore, $x^*(N) \le \CoS(\cal G)$; in particular, if $\core(\cal G) \ne \emptyset$ then $x^*(N) \le v(N)$.
\end{theorem}
\begin{proof}
Let us first consider the problem of computing the cost of stability. It can be expressed as the following linear optimization problem:
\begin{align}
\text{min } & \sum_{i = 1}^n x_i & \label{lp:CoS}\\
\text{s.t. } & \sum_{i \in S} x_i \ge v(S) & \forall S \subseteq N \notag
\end{align}
If the value of the optimal solution to \eqref{lp:CoS} is at most $v(N)$ then the core of $\cal G$ is not empty. Now, clearly \eqref{lp:CoS} is not poly-time computable, since the number of constraints is exponential. However, given $m$ i.i.d. samples $(S_1,v(S_1)),\dots,(S_m,v(S_m))$, consider the following optimization problem:
\begin{align}
\text{min } & \sum_{i = 1}^n x_i & \label{lp:PAC-stable}\\
\text{s.t. } & \sum_{i \in S_j} x_i \ge v(S_j) & \forall j \in \{1,\dots,m\}\notag
\end{align}
Unlike \eqref{lp:CoS}, \eqref{lp:PAC-stable} is poly-time computable in $m$ and $n$. Moreover, if $\vec x^*$ is the optimal solution to \eqref{lp:PAC-stable}, then the value of $x^*(N)$ is no greater than the value of the optimal solution to \eqref{lp:CoS}, as \eqref{lp:CoS} imposes more constraints on the target function than \eqref{lp:PAC-stable}. To see why \eqref{lp:PAC-stable} outputs a payoff division that PAC stabilizes $\cal G$, we observe that the problem of finding a PAC stable payoff division is equivalent to the problem of learning an unknown linear function $\vec x$ such that $\sum_{i \in S}x_i \ge v(S)$ for all $S \subseteq N$: since \eqref{lp:PAC-stable} is a consistent algorithm, we know that it outputs a PAC approximation of $\vec x$.
\end{proof}
Theorem~\ref{thm:pac-stable} states that any cooperative game is PAC stabilizable, irrespective of its own learnability guarantees. 
While the proof of Theorem~\ref{thm:pac-stable} is simple, it has several important implications. The immediate implication of Theorem~\ref{thm:pac-stable} is overcoming the computational complexity of finding stable outcomes. If one is willing to forgo guaranteed stability, it is possible to find payoff divisions that are likely to be stable. 
Second, PAC stabilizability is a stability concept that is founded on observational data: one does not need to know anything about the underlying cooperative game in order to use partial observations of its values to achieve stability. 
Finally, the underlying assumptions about the PAC stabilizability concept are reasonable from a practical perspective. There are several works that study the stability of cooperative games when certain restrictions on the coalitions that may form are in place. The most notable example of such a line of work are Myerson interaction graphs~\cite{myerson1977graphs}; in addition to the game $\cal G = \tup{N,v}$, we are given a connected, undirected graph $\Gamma = \tup{N,E}$, whose nodes are the players. A coalition $S \subseteq N$ may form only if the subgraph induced by $S$  in $\Gamma$ is connected. Several works have shown that under certain assumptions on the structure of the Myerson interaction graph, graph restricted coalitional games may be stabilized~\cite{demange2004}, exhibit a low cost of stability~\cite{bousquet2015graphs,meir2013cost}, and have stable outcomes found in polynomial time~\cite{chalkiadakis2012stability,zick2012comp}. Rather than assuming a certain underlying structure on the game, the PAC stabilizability approach observes the actual coalitions formed, and assumes that past events are a good prediction of future agent coalition formation habits; thus, if all observed coalitions were, say, of size at most 2, it is likely that the PAC stable payoff division would be stable against pairwise deviations. 

\section{PAC Learnability of Common Classes of Cooperative Games}
\label{sec:common}
In what follows, we explore the PAC learnability of common classes of cooperative games. Some of our computational intractability results depend on the assumption that $\NP \ne \RP$, where $\RP$ is the class of all languages for which there exists a poly-time algorithm that for every instance $I$, outputs ``no'' if $I$ is a no instance, and ``yes'' with probability $\ge \frac12$ if it is a ``yes'' instance. It is believed that $\NP \ne \RP$~\citep{hemaspaandra2002complexity}. 

\subsection{Network Flow Games}
A {\em network flow game} is given by a weighted, directed graph $\Gamma = \tup{V,E}$, with $w:E\to \R_+$ being the weight function for the edges. Here, $N = E$, and $v(S) = \maxflow(\Gamma|_S,w,s,t)$, where $\maxflow$ denotes the maximum $s$-$t$ flow through $\Gamma$, where edge weights are given by $w$, and $s,t \in V$.

We begin by showing that a similar class of functions is not efficiently learnable. We define the following family of functions, called {\em min-sum} functions which are defined as follows: there exists a list of $n$-dimensional, non-negative integer vectors $\vec w_1,\dots,\vec w_k$. For every $S \subseteq N$, $f(S) = \min_{\ell\in[k]} \vec w_\ell(S)$, where $\vec w_\ell(S) = \sum_{j\in S} w_{\ell j}$. If $k=1$, we say that the min-sum function is \emph{trivial}. We note that \cite{balcan2011learning} study the learnability of XOS valuations, where the $\min$ is replaced with a $\max$. 

We define $k$-min-sum to be the class of min-sum functions defined with $k$ vectors.

\begin{lemma}\label{lem:min-sum-small-trees-unlearnable}
The set of $k$-min-sum functions is not efficiently PAC learnable unless $\NP = \RP$ whenever $k \ge 3$.
\end{lemma}
\begin{proof}
Our proof relies on the fact that CNF formulas with more than two clauses are not efficiently learnable unless $\NP = \RP$~\citep{pitt1988computational}.

Given a set of variables $x_1,\dots,x_n$, let us define a set of players 
$$N = \{x_1,\bar x_1,\dots,x_n,\bar x_n,y\}.$$
Given a $k$-clause CNF formula of the form $\phi = \bigwedge_{\ell = 1}^k C_\ell$, where $C_\ell$ is a disjunctive clause containing literals from $N$ (excluding the variable $y$), we define the following $k+1$-min-sum function $f_\phi:2^N \to \{0,1\}$:
$$f_\phi(S) = \min\{(|C_j \cap S|)_{j = 1}^m,|S\cap \{y\}|\}.$$
In order to have a value of 1, $S$ must intersect with every $C_j$ on at least one player; otherwise, $f_\phi(S) = 0$. Moreover, $S$ must contain $y$. 
In what follows, we will take truth assignments on $x_1,\dots,x_n$ and map them to subsets of players in $N$, ensuring that $f_\phi(S) = 1$ if and only if the truth assignment from which $S$ was generated from an assignment satisfying $\phi$.  
We note that it is possible that one cannot generate a truth assignment for $\phi$ from all sets $S$ for which $f_\phi(S) = 1$; for example, $f_\phi(N) = 1$, but this is completely uninformative. Given a truth assignment $T$ for $\phi$, we define its set equivalent to be $S_T$, where $S_T$ contains $x_i$ if $x_i$ is true in $T$, otherwise $S_T$ contains $\bar x_i$. Also, $S_T$ contains $y$. Thus, $f_\phi(S_T) = 1$ if and only if $T$ satisfies $\phi$. 

Now, given a set of inputs from $\phi$ $(T_1,\phi(T_1)),\dots,(T_m,\phi(T_m))$, we write $\cal T = \{T_1,\dots,T_m\}$. 
For every $T \in \cal T$, we add to the input $(S_T,\phi(T))$, and $(S_T\setminus\{y\},0)$. The sampled point $(S_T\setminus \{y\},0)$ is added to ensure that the ``importance'' of $y$ in the definition of $f_\phi$ is noted by any algorithm that is consistent with the input. In other words, we can ``pretend'' that the input of truth assignments to $\phi$ is an input of an unknown $k+1$-min-sum function as defined above, and use it to define a CNF formula that is consistent with $\phi$. Moreover, if the original input $m$ truth assignments, the number of points we provide to the algorithm learning $f_\phi$ is $2m$.

Suppose that there exists some consistent poly-time algorithm $\cal A$ for $k+1$-min-sum functions; that is, given a list of $m$ inputs from an unknown $k+1$-min-sum function $f$, $\cal A$ outputs a list of non-negative vectors $\vec w_1^*,\dots,\vec w_{k+1}^*$ that define a $k+1$-min-sum function that is consistent on the inputs, and does so in time polynomial in $n$,$\frac1\eps$, and $\log\frac1\delta$. Then, if we input to this algorithm the inputs defined above, it will output $f^*$ defined by $\vec w_1^*,\dots,\vec w_{k+1}^*$, such that $f^*(S_{T_\ell}) = f_\phi(S_{T_\ell})$ on the inputs we designed.
Using $f^*$, we now show how one can reconstruct a CNF formula with at most $k$ clauses, such that $\phi^*(T) = \phi(T)$ for all $T \in \cal T$.

Let us define $\cal T_-$ to be the set of truth assignments in $\cal T$ that do not satisfy $\phi$ and $\cal T_+$ to be the set of truth assignments in $\cal T$ that do. We assume that we observe both satisfying and non-satisfying truth assignments (otherwise we can just output a trivial always true or always false CNF). 

First, we claim that there must exist at least one weight vector in the description of $f^*$ with the value of $y$ set to some positive quantity. Suppose that all of the vectors $\vec w_1^*,\dots,\vec w_{k+1}^*$ have the value of $y$ set to 0; then for any truth assignment $T \in \cal T$ that satisfies $\phi$, we have $(S_T,1)$ and $(S_T\setminus\{y\},0)$; but since none of the weight vectors of $f^*$ assign a positive weight to $y$, $f^*(S_T) = f^*(S_T\setminus\{y\})$, a contradiction to the consistency of $f^*$. We conclude that there exists at least one vector $\vec w_\ell^*$ that has a positive weight assigned to $y$.

Furthermore, we claim that there must exist at least one vector in the description of $f^*$ with the weight of $y$ set to $0$. This holds since for every $T \in \cal T_-$, there must be at least one vector $\vec w_\ell^*$ that takes a value of $0$ on $S_T$, to maintain consistency: vectors with a positive weight on $y$ have a positive weight for $S_T$, and if all of the vectors defining $f^*$ assign $S_T$ a positive weight, then $f^*(S_T)=1$, a contradiction to consistency.

Let us take the set of vectors who assign a weight of 0 to $y$, call that set $\cal W^*$; since there exist some vectors that assign a positive weight to $y$, $|\cal W^*|\le k$. Furthermore, since there exist some vectors who assign a weight of $0$ to $y$, $\cal W^*\ne \emptyset$. For every $\vec w_\ell^* \in \cal W^*$, we define a clause $C_\ell^*$ such that $C_\ell^*$ is a disjunction of all literals in the support of $\vec w_\ell^*$; in other words, if the weight of $x_i$ in $\vec w_\ell^*$ is positive, then $x_i$ is in $C_\ell^*$, and if the weight of $\bar x_i$ is positive then $\bar x_i$ is in $C_\ell^*$. 

Let us write the resulting CNF formula to be $\phi^*$. First, since $|\cal W^*|\le k$, the number of clauses in $\phi^*$ is at most $k$. Second, if $\phi(T) = 1$, then $f^*(S_T) = 1$, and in particular, $w_\ell^*(S_T) > 0$ for all $\vec w_\ell^* \in \cal W^*$; thus, $\phi^*(T) = 1$ as well. Finally, if $\phi(T) = 0$, then $f(S_T) = 0$, and there is at least one weight vector in $\cal W^*$ that has $w_\ell^*(S_T) = 0$. The corresponding clause $C_\ell^*$ is not satisfied by $T$, and in particular, $\phi^*(T) = 0$. This concludes the proof.
\end{proof}
\begin{theorem}\label{cor:networkflows-unlearnable}
Network flow functions are not efficiently learnable unless  $\NP = \RP$.
\end{theorem}
\begin{proof}
Our proof reduces the problem of learning min-sum functions to the problem of learning network flow functions. Given a min-sum target function $f$, defined by $\vec w_1,\dots,\vec w_k$ (with $k \ge 3$), and a distribution $\cal D$ over samples of $N$, we construct the directed graph $\Gamma = \tup{V,E}$ as follows (see Figure~\ref{fig:networkflows-unlearnable}). 

For every weight vector $\vec w_\ell = (w_{1\ell},\dots,w_{n\ell})$, we define vertices $\ell,\ell+1$, and $n$ edges from $\ell$ to $\ell+1$, where the capacity of the edge $e_{i\ell}$ is $w_{i\ell}$. Finally, we denote the vertex $k+1$ as the target $t$, and the vertex $1$ as the source $s$. Given a set $S \subseteq N$, we write $E_S = \{e_{i\ell}\mid \ell \in [k],i \in S\}$. We observe that the flow from $s$ to $t$ in the constructed graph using only edges in $E_S$ equals $f(S)$; in other words, $\maxflow_\Gamma(E_S) = f(S)$ for all $S \subseteq N$. 
Now, given a probability distribution $\cal D$ over $2^N$, we define a probability distribution over $E$ as follows: $\Pr_{\cal D'}[E_S] = \Pr_{\cal D}[S]$ for all $S \subseteq N$, and is 0 for all other subsets of $E$.

We conclude that efficiently PAC learning $\maxflow_\Gamma$ under the distribution $\cal D'$ is equivalent to PAC learning $f$, which cannot be done efficiently by Lemma~\ref{lem:min-sum-small-trees-unlearnable}.
\tikzstyle{simplenode}=[shape=circle,draw, fill=blue!20, minimum size=1em]
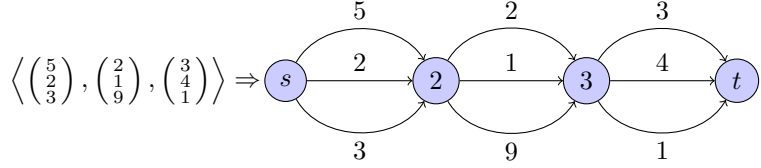
\begin{figure}
	\centering
	\begin{tikzpicture}[node distance=2cm]
	\node(mathpart){$\left\langle\left(\begin{smallmatrix}5\\2\\3\end{smallmatrix}\right),\left(\begin{smallmatrix}2\\1\\9\end{smallmatrix}\right),\left(\begin{smallmatrix}3\\4\\1\end{smallmatrix}\right)\right\rangle \Rightarrow$};
	\node(s)[simplenode,right of =mathpart]{$s$};
	\node(n1)[simplenode,right of = s]{$2$};
	\node(n2)[simplenode,right of =n1]{$3$};
	\node(t)[simplenode,right of = n2]{$t$};
	\path[->] (s) edge[bend left=60] node[above]{$5$} (n1);
	\path[->] (s) edge[bend left=-60] node[below]{$3$} (n1);
	\path[->] (s) edge node[above]{$2$} (n1);
	\path[->] (n1) edge[bend left=60] node[above]{$2$} (n2);
	\path[->] (n1) edge[bend left=-60] node[below]{$9$} (n2);
	\path[->] (n1) edge node[above]{$1$} (n2);
	\path[->] (n2) edge[bend left=60] node[above]{$3$} (t);
	\path[->] (n2) edge[bend left=-60] node[below]{$1$} (t);
	\path[->] (n2) edge node[above]{$4$} (t);
	\end{tikzpicture}
	\caption{An example of the reduction from $k$-min-sum functions to network flow games, described in Theorem~\ref{cor:networkflows-unlearnable}. Note that we rename the nodes $1$ and $4$ to $s$ and $t$, respectively. The description of the original $3$-min-sum function is given in $\langle \cdot \rangle$.}
	\label{fig:networkflows-unlearnable}
\end{figure}
\end{proof}

Learning network flow games is thus generally a difficult task, computationally speaking. In order to obtain some notion of tractability, let us study a learning scenario, where we limit our attention to sets that constitute paths in $\Gamma$. In other words, we limit our attention to distributions $\cal D$ such that if $\cal D$ assigns some positive probability to a set $S$, then $S$ must be an $s$-$t$ path in $\Gamma$. One natural example of such a distribution is the following: we make graph queries on $\Gamma$ by performing a random walk on $\Gamma$ until we either reach $t$ or have traversed more than $|V|$ vertices. 

Given a directed path $p = (w_1,\dots,w_k)$, we let $w(p)$ be the flow that can pass through $p$; that is, $w(p) = \min_{e \in p} w_e$. 
\begin{theorem}\label{thm:network-flow-games-path-learnable}
Network flow games are efficiently PAC learnable if we limit $\cal D$ to be a distribution over paths in $\Gamma$.
\end{theorem}
\begin{proof}
Given an input $((p_1,v_1),\dots,(p_m,v_m))$, we let $\bar w_e = \max_{j:e \in p_j}v_j$. 

We observe that the weights $(\bar w_e)_{e \in E}$ are such that $\bar w(p_j) = w(p_j)$ for all $j \in [m]$. This is because for any $e \in p_j$, $\bar w_e \ge v_j$, so $\min_{e \in p_j} \bar w_e \ge v_j$. On the other hand, $\bar w_e \le w_e$ for all $e \in E$, since $w_e \ge v_j$ for all $v_j$ such that $e \in p_j$, and in particular $w_e \ge \max_{j:e \in p_j} v_j = \bar w_e$. Thus, $\min_{e \in p_j} \bar w_e\le \min_{e \in p_j}w_e = v_j$. In other words, by simply taking edge weights to be the maximum flow that passes through them in the samples, we obtain a graph that is consistent with the sample in polynomial time.

Now, suppose that the set of weights on the edges of the graph according to the target weights $w_e$ is given by $\{a_1,\dots,a_k\}$, where $k\leq n$. Then there are $(k+1)^n \le (n+1)^n$ possible ways of assigning values $(\bar w_e)_{e \in E}$ to the edges in $E$. In other words, there are at most $(n+1)^n$ possible hypotheses to test. Thus, in order to $(\eps,\delta)$-learn $(w_e)_{e\in E}$, where the hypothesis class $\cal C$ is of size $\le (n+1)^n$, we need a number of samples polynomial in $\frac 1 \eps, \log \frac1\delta$ and $\log|\cal C| \in \cal O(n\log n)$. 
\end{proof}
\subsection{Threshold Task Games}
In {\em Threshold task games (TTG)} each player $i \in N$ has an integer weight $w_i$; there is a finite list of tasks $\cal T$, and each task $t \in \cal T$ is associated with a threshold $q(t)$ and a payoff $V(t)$. Given a coalition $S \subseteq N$, we let $\cal T|_S = \{t \in \cal T\mid q(t) \le w(S)\}$. The value of $S$ is given by 
$v(S) = \max_{t \in \cal T|_S} V(t)$.
In other words, $v(S)$ is the value of the most valuable task that $S$ can accomplish. 
Weighted voting games (WVGs) are the special case of TTGs with a single task, whose value is 1; that is, they describe linear classifiers. 

Without loss of generality we can assume that all tasks in $\cal T$ have strictly monotone thresholds and values: if $q(t) > q(t')$ then $V(t) > V(t')$. Otherwise, we will have some redundant tasks. For ease of exposition, we assume that there is some task whose value is 0 and whose threshold is 0.
Let $\TTG(Q)$ be the class of $k$-TTGs for which the set of task values $Q\subseteq \R$ of size $k$. The first step of our proof is to show that $\TTG^k(Q)$ is PAC learnable.
\begin{lemma}\label{lem:ttg-known-tasks}
The class $\TTG^k(Q)$ is PAC learnable
\end{lemma}
\begin{proof}
In order to show this, we first bound the sample complexity of $\TTG^k(Q)$. We claim that $\pseudodim(\TTG^k(Q)) < (k+1)(n+2)$. The proof relies on the fact that the VC dimension of linear functions is $n+1$. 

Assume by contradiction that there exists some $\cal S$ of size $L$, where $L = (k+1)(n+2)$, and some values $r_1,\dots,r_L\in \R_+$ such that for all $\vec b \in \{0,1\}^L$ there is some TTG $f_{\vec b} \in \TTG^k(Q)$ such that $f_{\vec b}(S_j) \ge r_j$ when $b_j = 1$, and $f_{\vec b}(S_j) < r_j$ when $b_j = 0$. We assume that $0\le r_1\le\dots\le r_L$. 
Let us write the task set to $\cal T = \{t_1,\dots,t_k\}$, each with a value $V(t_\ell)\in Q$; we order our tasks by increasing value. 
Now, by the pigeonhole principle, there exists some task $t_\ell$ and some $j^*$ such that $r_{j^*},\dots,r_{j^*+(n+1)} \in [V(t_\ell),V(t_{\ell+1}))$. 
In particular, if we write $\cal S^* = \{S_{j^*},\dots,S_{j^*+(n+1)}\}$, then for every $\vec b \in \{0,1\}^L$, there is some $f_{\vec b} \in \TTG^k(Q)$ (defined by an agent weight vector $\vec w_{\vec b}$, and task thresholds $T_1^{\vec b},\dots,T_k^{\vec b}$), such that for all $S_j \in \cal S^*$, if $f_{\vec b}(S_j) > r_j$ it must be that $f_{\vec b}(S_j) > V_{\ell}$, i.e., $\vec w_{\vec b}(S_j) > T_\ell^{\vec b}$. If $f_{\vec b}(S_j) \le r_j$ then $\vec w_{\vec b}(S_j) \le T_\ell^{\vec b}$. Thus, $(\tup{\vec w_{\vec b},T_\ell^{\vec b}})_{\vec b}$ is a set of $n$-dimensional linear classifiers that is able to shatter a set of size $n+2$, a contradiction. 
To conclude, $\pseudodim(\TTG^k(Q)) \le (k+1)(n+2)$, which implies that the sample complexity for PAC learning TTGs is polynomial.  

It is easy to construct an efficient algorithm that is consistent with any sample from $\TTG^k(Q)$ via linear programming.

Given the inputs, 
$$(S_1,v_1),\dots,(S_m,v_m),$$ 
let us write the distinct values $\alpha_1,\dots,\alpha_\ell$ in $v_1,\dots,v_m$, 
and create $\ell$ tasks with values $V(t_1) = \alpha_1,\dots,V(t_\ell) = \alpha_\ell$. 
We observe that since $(S_1,v_1),\dots,(S_m,v_m)$ represent outputs of a function in $\TTG^k(Q)$, 
it must be the case that $\ell \le k$. 
We further assume that $V(t_1)<V(t_2)<\dots<V(t_\ell)$. We also define $t_{\ell+1}$ to be an auxiliary task that has $q(t_{\ell+1}) = V(t_{\ell+1}) = \infty$. 
Next, we obtain weights for the players and thresholds for the tasks. For every set $S_j$ it must be the case that if $v_j = V(t_r)$, then $w(S_j) \ge q(t_r)$, but $S_j$ does not have sufficient weight to complete $t_{r+1}$. Let $\sigma:[m] \to [\ell]$ be the mapping that, for each sample $(S_j,v_j)$, maps $S_j$ to the task that it completed; i.e. the task $t_r$ for which $v_j = V(t_r)$. 
In order to find the correct set of weights for tasks and agents we use linear programming; we require that every set $S_j$ has a weight of at least $q(t_{\sigma(j)})$, but no more than $q(t_{\sigma(j)+1})$ (we add a ``dummy'' variable $q(t_{\ell+1})$ whose weight is extremely high to make this requirement hold for the $\ell$-th task). Since we cannot explicitly code the condition $w(S_j)<q(t_{\sigma(j)+1})$ into an LP, we add a tolerance parameter $2^{-r}$. We repeatedly run LP~\eqref{lp:ttg} with $r =0,1,2,\dots$ until a feasible solution is found. Thus, we will need to rerun LP~\eqref{lp:ttg} at most the number of bits required to represent the values in the original TTG. 
\begin{align}
\mbox{find: } & \vec w \in\R_+^n,\vec q \in\R_+^\ell&\label{lp:ttg}\\
\mbox{s.t.: }& w(S_j) \ge q(t_{\sigma(j)}) & \forall j \in [m]\notag\\
& w(S_j) \le q(t_{\sigma(j)+1})+2^{-r}& \forall j \in [m]\notag\\
&w_i \ge 0 & \forall i \in [n]\notag\\
&q_j \ge 0 & \forall j \in [\ell]\notag
\end{align}
The linear feasibility program~\eqref{lp:ttg} has $n+ \ell$ variables and $2m$ constraints, and is thus solvable in polynomial time. 
Moreover, a feasible solution exists; namely, the one that corresponds to the weights in the original TTG. Thus, there is an efficient, consistent algorithm for $\TTG^k(Q)$. 
\end{proof}

Let $\TTG^k$ be the class of TTGs with $k$ tasks; the following lemma shows that if we take a sufficient number of samples, a game $v\in \TTG^k$ can be PAC approximated by a game $\bar v \in \TTG^k(Q)$, where $Q$ are the observed values of $v$.
\begin{lemma}\label{lem:ttg-similar-vals}
Given $m\ge k\frac1\eps\log\frac1\delta$ independent samples 
$$(S_1,v(S_1)),\dots,(S_m,v(S_m))$$ 
from $v \in \TTG^k$; 
let $Q = \bigcup_{j = 1}^m \{v(S_j)\}$. The event $\Pr_{S\sim \cal D}[v(S) \notin Q] < \eps$ occurs with probability at most $1 - \delta$.
\end{lemma}
\begin{proof}
First, recall that for every TTG in $\TTG^k$, we have $k$ different task values $$\{V_1,\dots,V_k\},$$ 
and any set of observed samples will show some $Q \subseteq \{V_1,\dots,V_k\}$; thus, there can be at most $2^k$ observed sets of values from the samples. 
Let us write $\cal D^m$ to be the probability distribution from which our samples are taken. Given $Q = \bigcup_{j = 1}^m \{v(S_j)\}$, let $Y_\eps(Q)$ be the event that $\Pr_{S\sim \cal D}[v(S)\notin Q] < \eps$. We need to bound the probability (over samples from $\cal D^m$) that the event $\lnot Y_\eps(Q)$ occurs.
Let $\cal S_m(Q)$ be the set of all sequences sampled from $\cal D^m$ for which the observed set of values is $Q$.

First, we note that if $S_1,\dots,S_m$ generate the values $Q$, such that $\lnot Y_\eps(Q)$ occurs, then $\Pr_{S\sim \cal D}[v(S) \in Q] \le 1 - \eps$. Next, note that if $(S_1,\dots,S_m) \in \cal S_m(Q)$ then $v(S_j) \in Q$ for all $j = 1,\dots,m$. Thus, 
\begin{align*}
\Pr[(S_1,\dots,S_m) \in \cal S_m(Q)] \le& \prod_{j =1}^m \Pr_{S\sim \cal D}[v(S) \in Q] \le (1 -\eps)^m,
\end{align*}
which by our choice of $m$, is at most $\frac{\delta}{2^k}$:
\begin{align*}
(1 - \eps)^m &\le \euler^{-\eps m} \le \euler^{-\eps k\frac1\eps \log\frac1\delta}\\
						 &=\frac{\delta}{\euler^k} < \frac{\delta}{2^k}
\end{align*}
Putting it all together,
\begin{align*}
\Pr[\lnot Y_\eps(Q)] =& \sum_{Q: \lnot Y_\eps(Q)}\Pr[(S_1,\dots,S_m) \in \cal S_m(Q)] \\
<& 2^k\frac{\delta}{2^k} = \delta
\end{align*}
which concludes the proof.
\end{proof}
Using the two lemmas, we are now ready to prove that the class of $k$-TTGs, $\TTG^k$, is PAC learnable.
\begin{theorem}\label{thm:ttg}
Let $\TTG^k$ be the class of $k$-TTGs; then $\TTG^k$ is PAC learnable.
\end{theorem}
\begin{proof}[Proof Sketch]
Let $(S_1,v(S_1)),\dots,(S_m,v(S_m))$ be our set of samples. According to Lemma~\ref{lem:ttg-similar-vals}, we can choose $m$ such that with probability $\ge 1 - \frac\delta2$, $\Pr_{S\sim \cal D}[v(S)\notin Q] < \frac\eps2$. 
We let $\bar v$ be the TTG $v$ with the set of tasks reduced to $Q$; that is $\bar v(S) = v(S)$ if $v(S) \in Q$, and is the value of the best task that $S$ can complete whose value is in $Q$ otherwise. Thus, we can pretend that our input is from $\bar v \in \TTG^k(Q)$. According to Lemma~\ref{lem:ttg-known-tasks}, if $m$ is sufficiently large, then with probability $\ge 1-\frac\delta2$ we will output some $v^* \in \TTG^k(Q)$ such that $\Pr_{S\sim \cal D}[\bar v(S) = v^*(S)] \ge 1 - \frac\eps2$.
Thus, with probability $\ge 1 - \delta$, we have that both $\Pr_{S\sim \cal D}[\bar v(S) = v^*(S)] \ge 1 - \frac\eps2$ and $\Pr_{S\sim \cal D}[v(S) = \bar v(S)] \ge 1 - \frac\eps2$. We claim that $v^*$ PAC approximates $v$. Indeed, 

\begin{align*}
\Pr_{S\sim \cal D}[v(S) \ne v^*(S)] =& \Pr_{S\sim \cal D}[(v(S) \ne v^*(S))\land(v(S) = \bar v(S))]\\
&+ \Pr_{S\sim \cal D}[(v(S) \ne v^*(S))\land(v(S) \ne \bar v(S))]\\
\le& \Pr_{S\sim \cal D}[v^*(S) \ne \bar v(S)] + \Pr_{S\sim \cal D}[v(S) \ne \bar v(S)] < \eps
\end{align*}
\end{proof}

\subsection{Induced Subgraph Games}
An induced subgraph game (ISG)~\cite{deng94complexity} is given by a weighted graph $\Gamma = \tup{N,E}$, where for every pair $i,j \in N$, $w_{ij}\in \Z$ denotes the weight of the edge between $i$ and $j$. We let $W$ be the weighted adjacency matrix of $\Gamma$. The value of a coalition $S \subseteq N$ is given by 
$v(S) = \sum_{i \in S}\sum_{j \in S\mid j>i}w_{ij}$;
i.e. the value of a set of nodes is the weight of their induced subgraph.

\begin{theorem}
The class of induced subgraph games is efficiently PAC learnable.
\end{theorem}
\begin{proof}
Let $W$ be the (unknown) weighted adjacency matrix of $\Gamma$. Let us write $\vec e_S$ to be the indicator vector for the set $S$ in $\R^n$. That is, the $i$-th coordinate of $\vec e_S$ is 1 if $i \in S$, and is 0 otherwise. We observe that in an ISG, $v(S) = \vec e_S^T W \vec e_S$. 
In other words, learning the coefficients of an ISG is equivalent to learning a linear function with $\cal O(n^2)$ variables (one per vertex pair), which is known to have polynomial sample complexity~\citep{anthony2009neural}.

Now, given observations $(S_1,v_1),\dots,(S_m,v_m)$, we need to solve a linear system with $m$ constraints (one per sample), and $\cal O(n^2)$ variables (one per vertex pair, as above), which is solvable in polynomial time.
\begin{align}\label{eq:isg}
\mbox{Find:} & (w_{i,i'})_{i,i' \in N}&\\
\mbox{s.t.} & \sum_{i,i' \in S_j}w_{i,i'} = v_j& \forall j = 1,\dots,m\notag
\end{align} 
The output of~\eqref{eq:isg} is guaranteed to be consistent, and since a solution exists (namely, $W$), we have a straightforward consistent poly-time algorithm, and conclude that the class of ISGs is efficiently PAC learnable.
\end{proof}

\subsection{Additional Classes of Cooperative Games}
Before we conclude, we present a brief overview of additional results obtained for other classes of cooperative games. 
\paragraph{Vector Weighted Voting Games :}
In weighted voting games (WVGs), each player $i \in N$ has an integer weight $w_i$; the weight of a coalition $S \subseteq N$ is defined as $w(S) = \sum_{i \in S}w_i$. A coalition is {\em winning} (has value 1) if $w(S) \ge q$, and has a value of 0 otherwise. Here, $q$ is a given {\em threshold}, or {\em quota}. The class of vector WVGs is a simple generalization of weighted voting games given by \citename{elkind2009computational}. 
A vector WVG of degree $k$ (or $k$-vector WVG) is given by $k$ WVGs: $\tup{\vec w_1;q_1},\dots,\tup{\vec w_k;q_k}$. A set $S \subseteq N$ is winning if it is winning in every one of the $k$ WVGs. 

Learning a weighted voting game is equivalent to learning a separating hyperplane, which is known to be easy~\citep{kearns1994introduction}. However, learning $k$-vector WVGs is equivalent to learning the intersection of $k$-hyperplanes, which is known to be $\NP$-hard even when $k = 2$~\citep{alekhnovich2004learnability,blum1992training,klivans2002learning}. Thus, $k$-WVGs are not efficiently PAC learnable, unless $P=\NP$.

\paragraph{MC-nets: }
{\em Marginal Contribution Nets (MC-nets)}~\citep{ieong2005mcnets} provide compact representation for cooperative games. Briefly, an MC-net is given by a list of rules over the player set $N$, along with values. A rule is a Boolean formula $\phi_j$ over $N$, and a value $v_j$. For example, $r = x_1\land x_2 \land \lnot x_3 \to 7$ assigns a value of 7 to all coalitions containing players 1 and 2, but not player 3. Given a list of rules, the value of a coalition is the sum of all values of rules that apply to it. PAC learning MC-nets can be reduced to PAC learning of DNF formulas, which is known to be intractable~\citep{klivans2001learning}. 

More formally, let $\phi = \bigvee_{j=1}^m C_j$ be a DNF formula, where $C_1,\dots,C_m$ are conjunctive clauses over a set of $n$ variables. The reduction is somewhat similar to the one used in the proof of Theorem~\ref{cor:networkflows-unlearnable}: 
given the set of variables $x_1,\dots,x_n$, we define our player set to be $N = \{1,\dots,n\}$. We perform the following transformation: for every clause $C_j$ in the DNF we say that a set $S\subseteq N$ satisfies the clause if $i \in S$ whenever $x_i \in C_j$, and $i \notin S$ if $\lnot x_i \in C_j$. We define a rule of the form $C_j \rightarrow 1$; that is, if $S\subseteq N$ satisfies the clause it is awarded one point. We can associate a truth assignment for the variables $x_1,\dots,x_n$ with a set in $N$ in the natural way: if $i\in S$ iff $x_i$ is set to true. Thus, if a truth assignment satisfies the DNF $\phi$, it has a positive value under the MC net; otherwise, its value is $0$. To conclude, any algorithm that properly learns MC-nets will be easily transformed into one that properly learns DNF formulas. 

\paragraph{Coalitional Skill Games:}
Coalitional Skill Games (CSGs)~\citep{bachrach2008skill} are another well-studied class of cooperative games. Here, each player $i$ has a skill-set $K_i$; additionally, there is a list of tasks $\cal T$, each with a set of required skills $\kappa_t$. Given a set of players $S\subseteq N$, let $K(S)$ be the set of skills that the players in $S$ have. Let $\cal T(S)$ be the set of tasks $\{t \in \cal T\mid \kappa_t \subseteq K(S)\}$. The value of the set $\cal T(S)$ can be determined by various utility models; for example, setting $v(S) = |\cal T(S)|$, or assuming that there is some subset of tasks $\cal T^* \subseteq \cal T$ such that $v(S) = 1$ iff $\cal T^* \subseteq \cal T(S)$; the former class of CSGs is known as {\em conjunctive task skill games (CTSGs)}. 

PAC learnability of coalitional skill games is generally computationally hard. This holds even if we make some simplifying assumptions; for example, even if we know the set of tasks and their required skills in advance, or if we know the set of skills each player possesses, but the skills required by tasks are unknown. However, we can show that CTSGs are efficiently PAC learnable if player skills are known. 

\section{Discussion}

Our work is limited to finding outcomes that are likely to be stable for an unknown function. However, learning approximately stable outcomes is a promising research avenue as well. Such results naturally relate approximately stable outcomes --- such as the $\eps$ and least core~\citep{coopbook}, or the cost of stability~\citep{costab2009} --- with PMAC learning algorithms~\citep{balcan2011submodular}, which seek to \emph{approximate} a target function (M stands for ``mostly'') with high accuracy and confidence.

This work has focused on the core solution concept; however, learning other solution concepts is a natural extension. While some solution concepts, such as the nucleolus or the approximate core variants mentioned above, can be naturally extended to cases where only a subset of the coalitions is observed, it is less obvious how to extend solution concepts such as the Shapley value or Banzhaf power index. These concepts depend on the marginal contribution of player $i$ to coalition $S$, i.e., $v(S\cup\{i\})-v(S)$. Under the Shapely value, we are interested in the expected marginal contribution when a permutation of the players is drawn uniformly at random, and $i$ joins previous players in the permutation. According to Banzhaf, $S$ is drawn uniformly at random from all subsets that do not include $i$. Both solution concepts are easy to approximate if we are allowed to draw coalition values from the appropriate distribution~\citep{bachrach10approx} --- this is a good way to circumvent computational complexity when the game is known. It would be interesting to understand what guarantees we obtain for arbitrary distributions.

Finally, other models of cooperative behavior would naturally extend to a learning environment. For example, in {\em hedonic games}~\citep{aziz2016hedoniccomsoc}, each player reports a complete preference order over coalitions --- either ordinal, e.g. ``player $i$ prefers coalition $S$ to coalition $T$, or cardinal, i.e. each player assigns a numerical score to each coalition. Given players' preferences over coalitions, our goal is to find a {\em coalition structure} (i.e. a partition of players into groups), that satisfy certain notions of stability or fairness. One such solution concept is {\em Nash stability}; a partition $\pi$ is Nash stable, if no player can leave the coalition it belongs to under $\pi$, and benefit by moving to another existing group under $\pi$. The same line of investigation raised in this work can be applied to the hedonic setting: given a set of coalitions and some information about them (e.g. the value that their members assign to them, their order of preference according to the players etc.), can we find a partition of the agents that will likely satisfy certain fairness properties? 

\paragraph{Acknowledgments:}
This research is partially supported by NSF grants \\
CCF-1451177,
CCF-1422910, and CCF-1101283,  CCF-1215883 and IIS-1350598, a Microsoft Research Faculty Fellowship,
and Sloan Research Fellowships.

The authors would like to thank Amit Daniely for a useful discussion which led to Theorem~\ref{thm:pac-stable}, as well as the reviewers of the paper's original IJCAI 2015 version. 

\bibliographystyle{named}
\bibliography{bibshort}

\end{document}